\documentclass[
reprint,
superscriptaddress,
showpacs,
nofootinbib,
amsmath,amssymb,
aps,
pra,
floatfix,longbibliography
]{revtex4-2}
\usepackage[T1]{fontenc}
\usepackage{graphicx}
\usepackage{dcolumn}
\usepackage[hypertexnames, breaklinks=true, bookmarksnumbered=true, bookmarksopen=true, colorlinks=true, linktocpage=true, citecolor=blue, urlcolor=magenta, linkcolor=magenta]{hyperref}
\usepackage{amsfonts,amsthm}
\usepackage{mathrsfs}
\usepackage{bbm} 
\usepackage{bm} 
\usepackage{mathtools} 
\usepackage[stretch=10]{microtype} 
\usepackage[normalem]{ulem} 
\usepackage{empheq}
\usepackage{color}
\usepackage{tcolorbox}
\usepackage{natbib}
\usepackage{pxfonts,txfonts}
\usepackage{tgpagella}

\allowdisplaybreaks

\DeclareMathOperator{\tr}{Tr}

\newcommand{\ran}{\rangle}
\newcommand{\lan}{\langle}

\newcommand{\id}{\mathbbm{1}}

\newcommand{\bra}[1]{\langle #1|}
\newcommand{\ket}[1]{|#1\rangle}

\newcommand{\op}[2]{|#1\rangle \langle #2|}
\newcommand{\Tr}{\textrm{Tr}}

\newtheorem{theorem}{Theorem}
\newtheorem{lemma}{\emph{Lemma}}
\newtheorem{proposition}{Proposition}

\allowdisplaybreaks

\begin{document}


\title{Creating and destroying coherence with quantum channels}
\author{Masaya Takahashi}
\email{masaya.takahashi@siu.edu}
\affiliation{Department of Physics and Astronomy, Southern Illinois University, Carbondale, Illinois 62901, USA}
\author{Swapan Rana}
\affiliation{Centre for Quantum Optical Technologies, Centre of New Technologies,
University of Warsaw, Banacha 2c, 02-097 Warsaw, Poland}
\affiliation{Physics and Applied Mathematics Unit, Indian Statistical Institute, 203 B T Road, Kolkata 700108, India}
\author{Alexander Streltsov}
\affiliation{Centre for Quantum Optical Technologies, Centre of New Technologies,
University of Warsaw, Banacha 2c, 02-097 Warsaw, Poland}
\date{\today}

\begin{abstract} The emerging quantum technologies rely on our ability to establish and control quantum systems in nonclassical states, exhibiting entanglement and quantum coherence. It is thus crucial to understand how entanglement and coherence can be created in the most efficient way. In this work we study optimal ways to create a large amount of quantum coherence via quantum channels. For this, we compare different scenarios, where the channel is acting on an incoherent state, on states which have coherence, and also on subsystems of multipartite quantum states. We show that correlations in multipartite systems do not enhance the ability of a quantum channel to create coherence. We also study the ability of quantum channels to destroy coherence, proving that a channel can destroy more coherence when acting on a subsystem of a bipartite state. Crucially, we also show that the destroyed coherence on multipartite system can exceed the upper bound of those on the single system when the total state is entangled. Our results significantly simplify the evaluation of coherence generating capacity of quantum channels, which we also discuss.
\end{abstract}

\maketitle

\section{Introduction \label{Sec:Introduction}} 
The superposition principle of quantum mechanics is one of the main reasons for the discrepancy between classical and quantum physics. It leads to nonclassical phenomena, such as quantum coherence and entanglement, which can be used for quantum technological applications~\cite{Horodecki-2009a,Streltsov-2017a,Wu2105.06854}. An important example is quantum state merging~\cite{Horodecki-2005a}, where entanglement between remote parties can be used to merge their parts of a shared quantum state. Having access to local coherence allows to reduce the entanglement consumption in this task~\cite{Streltsov-2016a}.

In this work we investigate optimal ways for creating quantum coherence. For this we will apply tools from the resource theory of quantum coherence \cite{Aberg-2006a,Baumgratz-2014a,Winter-2016a,Streltsov-2017a}. In this theory, quantum states which are diagonal with respect to a reference basis $\{\ket{i}\}$ are called incoherent. Correspondingly, a quantum operation is called incoherent if it can be decomposed into Kraus operators $\{K_j\}$ which do not create coherence i.e., $K_j\ket{m}\propto \ket{n}$~\cite{Baumgratz-2014a}. Within the resource theory of coherence, the basic unit is the maximally coherent qubit state $\ket{+}=(\ket{0}+\ket{1})/\sqrt{2}$. By using this state -- or many copies thereof -- along with incoherent operations, it is possible to prepare an arbitrary quantum state $\rho$, and to implement an arbitrary transformation of a quantum system~\cite{Baumgratz-2014a,Chitambar-2016a,Ben-Dana-2017a,BenDanaPhysRevA.96.059903}. If instead one has access to a noisy state $\rho$, an incoherent distillation procedure can be applied to extract the state $\ket{+}$. The maximal rate of $\ket{+}$ states in the limit of many copies of $\rho$ is given by the relative entropy of coherence \cite{Winter-2016a}:
\begin{equation}
    C_r(\rho) = \min_{\sigma \in \mathcal{I}} S(\rho||\sigma) = S(\Delta\left[\rho\right]) - S(\rho). \label{eq:Cr}
\end{equation}
Here, $\mathcal I$ is the set of incoherent states, i.e., states which are diagonal in the reference basis $\{\ket{i}\}$. Moreover,
\begin{equation}
    S(\rho||\sigma) = \tr[\rho \log_2 \rho]- \tr[\rho \log_2 \sigma]
\end{equation}
is the quantum relative entropy, $S(\rho) = -\tr[\rho \log_2 \rho]$ is the von Neumann entropy, and $\Delta[\rho]=\sum \lan i|\rho|i\ran\,|i\ran\lan i|$ denotes complete dephasing in the reference basis $\{|i\ran\}$.

One way to create coherence is to apply a quantum channel $\Phi$ onto an incoherent state $\sigma$. The maximal amount of coherence achievable in this way is called \emph{cohering power} of $\Phi$ \cite{Baumgratz-2014a,Mani-2015a}: 
\begin{equation}
    \mathcal C(\Phi) = \sup_{\sigma \in \mathcal I} C\left(\Phi[\sigma]\right),
\end{equation}
where $C$ is a suitable coherence quantifier. Instead of applying the quantum channel $\Phi$ onto an incoherent state, it can be advantageous to have initial coherence to start with. We are then interested in the maximal increase of coherence, maximized over all quantum states:
\begin{equation}
    \mathscr{C}(\Phi) = \sup_\rho \left\{C\left(\Phi\left[\rho\right]\right) - C\left(\rho\right)\right\}.
\end{equation}
This quantity is known as the \emph{generalized cohering power}~\cite{Yao-2015a,Bu-2017b,Garcia-Diaz-2016a} of $\Phi$. Clearly generalized cohering power is never smaller than cohering power, which directly follows from their definitions. Also, if coherence is quantified via the relative entropy of coherence, there exist a channel $\Phi$ such that $\mathcal C(\Phi) < \mathscr{C}(\Phi)$~\cite{Bu-2017a}. The same is true for another quantifier of coherence based on the $\ell_1$-norm~\cite{Bu-2017a}, which will be discussed in more detail in Section~\ref{sec:MatrixNorms}.

In the most general case (see also Fig.~\ref{fig:ccp}), one can apply the quantum channel $\Phi$ onto one part of a bipartite quantum state $\rho^{AB}$, leading to the \emph{complete cohering power}~\cite{Ben-Dana-2017a}
\begin{equation}
\mathfrak{C}\left(\Phi\right) = \sup_k \left( \sup_{\rho^{AB}} \left\{ C\left(\Phi \otimes \mathbbm{1}_k \left[\rho^{AB}\right]\right) - C\left(\rho^{AB}\right) \right\} \right),
\end{equation}
where $k$ is Bob's local dimension. For a quantum channel $\Phi$, the action of the channel on an extended space is described by $\Phi \otimes \id_k$. Thus, complete cohering power of $\Phi$ is equal to the generalized cohering power of $\Phi \otimes \mathbbm{1}_k$ in the limit $k \rightarrow \infty$:
\begin{equation}
    \mathfrak{C}\left(\Phi\right) = \lim_{k \rightarrow \infty} \mathscr{C}(\Phi \otimes \mathbbm{1}_k).
\end{equation}
Hereafter, we will use $\mathbbm{1}$ for identity operation when dimension can be arbitrary and there is no confusion about it.

\begin{figure}
	\centering
	\includegraphics[width=0.9\columnwidth]{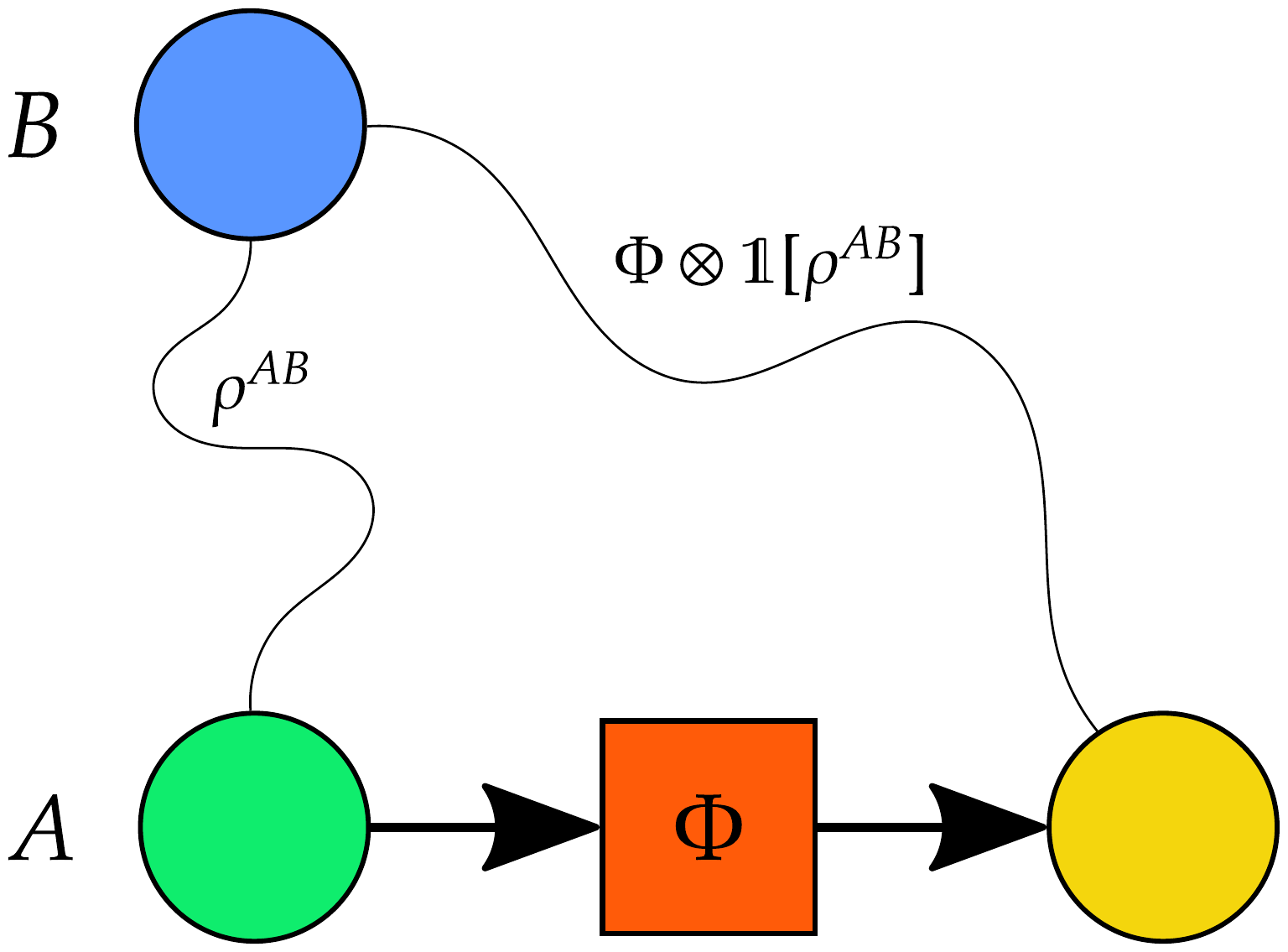}
	\caption{A quantum channel $\Phi$ can be used to create quantum coherence. In the most general setup, $\Phi$ can act on a part of a bipartite quantum state $\rho^{AB}$. The coherence created in this way is given by $C(\Phi\otimes \id[\rho^{AB}]) - C(\rho^{AB})$. As we show in this work, pre-established correlations are not useful for coherence generation. On the other hand, entanglement can enhance the ability of a channel to destroy quantum coherence.}
	\label{fig:ccp}
\end{figure}

While cohering power quantifies the ability of a quantum channel to create coherence, the \emph{decohering power} quantifies the maximal amount of coherence that the channel can destroy \cite{Mani-2015a}:
\begin{equation}
    \mathcal D(\Phi) = \sup_{\Phi[\rho]\, \in\, \mathcal I} C(\rho).
\end{equation}
In a similar way as for the cohering power, \emph{generalized} and \emph{complete decohering powers} are defined as \cite{Yao-2015a,Bu-2017a,Garcia-Diaz-2016a}
\begin{align}
    \mathscr{D}(\Phi) &= \sup_\rho \left\{ C(\rho) - C(\Phi[\rho]) \right\}, \\
    \mathfrak{D}(\Phi) &= \sup_{k} \left( \sup_{\rho^{AB}} \left\{ C\left(\rho^{AB}\right) - C\left(\Phi\otimes \id_k\left[\rho^{AB}\right]\right) \right\} \right).
\end{align}

So far we have defined cohering and decohering powers in a general fashion, without specifying the underlying coherence quantifier $C$. In the following, we will focus on the relative entropy of coherence introduced in Eq.~(\ref{eq:Cr}). Another coherence quantifier will be discussed in Section~\ref{sec:MatrixNorms}.

In general, pre-established correlations are expected to be useful for creating (or destroying) resources in physical processes. So it is interesting and important to explore the role of entanglement --- possibly the strongest from of correlations, in generating or eliminating coherence. Note that the relation between entanglement and coherence is well established at \emph{static} or state level~\cite{Streltsov-2015a,Streltsov-2016a,Chitambar-2016a,Takahashi-2018}, while our aim here is to explore the connection at a \emph{dynamic} level of operations.

\section{Properties of complete cohering power}

For a quantum system of dimension $d$, any maximally coherent state has the form $\ket{+_d}=\frac{1}{\sqrt{d}}\sum_{j}e^{i\theta_j}\ket{j}$. 
We will denote maximally coherent states also by $\ket{+}$, when dimension is arbitrary or obvious from the context. Since any state $\rho$ can be obtained from a maximally coherent state of the same dimension via incoherent operations~\cite{Baumgratz-2014a}, generalized cohering power is upper bounded by the coherence of the maximally coherent state regardless of coherence quantifier:
\begin{align*}
 \mathscr{C}(\Phi) \leq C \left(\op{+}{+}\right).
\end{align*}
Equality can be achieved with the channel $\Phi^{\max}$ having Kraus operators $K_i = \op{+}{i}$. 

We will now investigate the properties of complete cohering power. Since the dimension of ancillary system is in general not bounded, it is important to know if complete cohering power is a finite quantity. While the actual value of complete cohering power might differ for different channels, we are interested in general bounds, which only depend on the dimension of the system.
We start with the following lemma.
\begin{lemma} \label{lem:QuantumIncoherent}
	For any quantum-incoherent state $\rho^{AB} = \sum_i p_i \rho^A_i \otimes \op{i}{i}$, there is an index $n$ such that
	\begin{equation} C_r\left( \Phi \otimes \mathbbm{1} \left[ \rho^{AB} \right] \right) - C_r\left( \rho^{AB}\right) \leq  C_r\left(\Phi \left[\rho^A_n \right] \right) - C_r\left(\rho^A_n\right).\end{equation}
\end{lemma}
\begin{proof}
		Notice $C_r(\sum_i p_i \rho^A_i \otimes \op{i}{i}) = \sum_i p_i C_r(\rho^A_i)$ by simple calculation. Then we have
		\begin{align}
		 C_r\left( \Phi \otimes \mathbbm{1} \left[ \rho^{AB} \right] \right) - C_r\left( \rho^{AB}\right) 
		 =& C_r\left(\Phi \otimes \mathbbm{1}\left[\sum_i p_i \rho^A_i \otimes \op{i}{i}\right]\right) \nonumber \\
		 -& C_r\left(\sum_j p_j\rho^A_j \otimes \op{j}{j}\right) \nonumber \\
		 = & \sum_i p_i \left[C_r\left(\Phi\left(\rho^A_i\right)\right) - C_r\left(\rho^A_i\right)\right]. \label{eq:lemma1}
	    \end{align}
		 Since $0 \leq p_i \leq 1$ for all $i$, there $n$ satisfying
		 \begin{align}
		 \sum_i p_i \left[C_r\left(\Phi\left[\rho^A_i\right]\right) - C_r\left(\rho^A_i\right)\right] \leq C_r\left(\Phi\left[\rho^A_n\right]\right) - C_r\left(\rho^A_n\right).
		\end{align}
		Substituting this into Eq.~(\ref{eq:lemma1}), we finished the proof.
\end{proof}

Equipped with this Lemma, we are ready to prove the following theorem.
\begin{theorem}\label{thm:ccprel}
	Complete cohering power and generalized cohering power coincide:
	\begin{equation}
	    \mathscr{C}_r(\Phi) = \mathfrak{C}_r(\Phi).
	\end{equation}
\end{theorem}
\begin{proof}
		Our goal is to show 
		\begin{align}
		&\underset{\rho^{AB}}{\max} \left\{ C_r\left(\Phi \otimes \mathbbm{1}_k\left[\rho^{AB}\right]\right) - C_r\left(\rho^{AB}\right) \right\} \\
		&\leq  \underset{\rho^{A}}{\max} \left\{C_r\left(\Phi \left[\rho^{A}\right]\right) - C_r\left(\rho^A\right) \right\} \nonumber
		\end{align}
		for any $k$ and any quantum channel $\Phi$. The left hand side can be expanded as
		\begin{align}
		\underset{\rho^{AB}}{\max}& \left \{ C_r\left(\Phi \otimes \mathbbm{1}_k\left[\rho^{AB}\right]\right) - C_r\left(\rho^{AB}\right) \right\}\nonumber \\ 
		= \underset{\rho^{AB}}{\max}& \left\{S\left(\Phi \otimes \mathbbm{1}_k\left[\rho^{AB}\right] || \Delta\left[\Phi \otimes \mathbbm{1}_k\left[\rho^{AB}\right]\right]\right) - S\left(\rho^{AB} || \Delta \left[\rho^{AB}\right]\right) \right\} \nonumber \\
		= \underset{\rho^{AB}}{\max}& \bigg \{ S\left( \Phi \otimes \mathbbm{1}_k\left[\rho^{AB}\right] || \Delta^{B}\left[\Phi \otimes \mathbbm{1}_k\left[\rho^{AB}\right]\right]\right) \nonumber \\ 
		+&S\left(\Delta^B \left[\Phi \otimes \mathbbm{1}_k\left[\rho^{AB}\right]\right] || \Delta^{AB}\left[\Phi \otimes \mathbbm{1}_k\left[\rho^{AB}\right]\right]\right) \nonumber \\
		-&S\left(\rho^{AB} || \Delta^B\left[\rho^{AB}\right]\right) \nonumber \\
		-&S\left(\Delta^B\left[\rho^{AB}\right] || \Delta^{AB}\left[\rho^{AB}\right]\right) \bigg \} \label{eq:theorem1}
		\end{align}
		where $\Delta^B$ denotes dephasing in the incoherent basis of the system $B$ only.
		By monotonicity of the relative entropy under completely positive and trace-preserving maps we have
		\begin{equation}
		    S\left(\Phi\otimes\mathbbm{1}_{k}\left[\rho^{AB}\right]||\Delta^{B}\left[\Phi\otimes\mathbbm{1}_{k}\left[\rho^{AB}\right]\right]\right)-S\left(\rho^{AB}||\Delta^{B}\left[\rho^{AB}\right]\right)\leq 0.
		\end{equation}
		 From Eq.~(\ref{eq:theorem1}) then we get the following inequality:
		\begin{align*}
		\mathfrak{C}_r(\Phi) \leq \underset{\rho^{AB}}{\max} & \bigg\{ S\left(\Delta^B \left[\Phi \otimes \mathbbm{1}\left[\rho^{AB}\right]\right] || \Delta^{AB}\left[\Phi \otimes \mathbbm{1}\left[\rho^{AB}\right]\right]\right)  \\
		-&S\left(\Delta^B\left[\rho^{AB}\right] || \Delta^{AB}\left[\rho^{AB}\right]\right) \bigg \} \\
		= \underset{\rho^{AB}}{\max}& \left\{ C_r\left(\Delta^B \left[\Phi \otimes \mathbbm{1}\left[\rho^{AB}\right]\right] \right) - C_r\left(\Delta^B \left[\rho^{AB}\right]\right) \right \}.
		\end{align*}
	    The proof of the theorem is completed by using Lemma~\ref{lem:QuantumIncoherent}, and noting that $\Delta^B[\rho^{AB}]$ is a quantum-incoherent state.
\end{proof}

The above theorem shows that correlations with an ancillary system do not enhance the ability of a quantum channel to create quantum coherence. Since $C_r(\rho^A \otimes \sigma^B) =C_r(\rho^A) + C_r(\sigma^B)$, any state $\rho^A\otimes\sigma^B$ with any $\sigma^B$ is optimal for complete cohering power as long as $\rho^A$ is optimal for generalized cohering power. However, the other direction is not true in general: if $\rho^{AB}$ is an optimal input state for complete cohering power, $\Tr_B[\rho^{AB}]$ does not need to be optimal for the generalized cohering power. This can be verified with the input state $\ket{\phi^+} = (\ket{00} + \ket{11})/\sqrt{2}$ and $\Phi(\rho^A) = H\rho^AH^\dagger $ where $H$ is the Hadamard gate.

We now discuss applications of our results to the coherence generating capacity~\cite{Ben-Dana-2017a}. This quantity is defined in an operational way, assuming that a quantum channel $\Phi$ is applied $n$ times onto an incoherent state, with the goal to generate maximally coherent states~\cite{Ben-Dana-2017a}. In more detail, consider a quantum channel $\Phi:A \rightarrow A'$ and a sequence of bipartite incoherent operations $\mathcal{I}_i: A' \otimes B_{i-1} \rightarrow A \otimes B_{i}$. 
An initial incoherent state $\rho_0 \in A \otimes B_0$ is transformed by $\Phi \otimes \mathbbm{1}$ into $\rho'_0 = \Phi \otimes \mathbbm{1}[\rho_0]$. Then the incoherent operation $\mathcal{I}_1$ is applied on $\rho'_0$, leading to the state $\rho_1 =  \mathcal{I}_1[\rho'_0] = \mathcal{I}_1[\Phi \otimes \mathbbm{1}[\rho_0]]$. Iterating the procedure $n$ times we obtain a bipartite quantum state $\rho_n \in A \otimes B_n$. 
The coherence generating capacity of the channel $\Phi$ is now defined as~\cite{Ben-Dana-2017a}
\begin{equation}
    \mathbf{C}_{\textrm{gen}}(\Phi)=\sup\left\{ R:\lim_{n\rightarrow\infty}\left(\inf_{\left\{ \mathcal{I}_{i}\right\} }\left\Vert \rho_{n}-\ket{+_{2}}\!\bra{+_{2}}^{\otimes nR}\right\Vert _{1}\right)=0\right\},
\end{equation}
where $||M||_1 = \tr \sqrt{M^\dagger M}$ is the trace norm.

As was proven in~\cite{Ben-Dana-2017a}, complete cohering power of $\Phi$ is an upper bound on its coherence generating capacity. By using Theorem~\ref{thm:ccprel} the upper bound simplifies to
\begin{equation}
    \mathbf{C}_\mathrm{gen} (\Phi) \leq \max_\rho \left\{C_r\left(\Phi\left[\rho\right]\right) - C_r\left(\rho\right)\right\}.
\end{equation}
If we use maximally incoherent operations (i.e., all quantum operations which do not create coherence~\cite{Aberg-2006a}) between individual applications of the channel $\Phi$, the corresponding coherence generating capacity $\mathbf{C}_\mathrm{gen}^\mathrm{MIO}$ was shown to be equal to the complete cohering power~\cite{Garcia-Diaz-2018a}. With Theorem~\ref{thm:ccprel} we directly obtain
\begin{equation}
    \mathbf{C}_\mathrm{gen}^\mathrm{MIO} (\Phi) = \max_\rho \left\{C\left(\Phi\left[\rho\right]\right) - C\left(\rho\right)\right\}.
\end{equation}
Thus, our results significantly reduce the complexity to evaluate the coherence generating capacity of quantum channels.

\section{Properties of complete decohering power: advantage of entanglement}

We now explore the role of correlations in the process of destroying coherence for a given quantum channel $\Phi$. While correlations are not useful for creating coherence (see Theorem~\ref{thm:ccprel}), we will show that correlations between the system and an ancilla can enhance the ability of the channel to destroy quantum coherence. As an example, consider the erasing channel $\Lambda[\rho] = \op{0}{0}$. Let now $\Lambda$ act on one qubit of the two-qubit state 
\begin{equation}
    \ket{\psi^{AB}} = \frac{1}{\sqrt{2}}\left(\ket{0}\ket{+} + \ket{1}\ket{-}\right).
\end{equation}
Observe that $C_r(\op{\psi^{AB}}{\psi^{AB}}) = 2$ because the state is maximally coherent. The state after applying $\Lambda \otimes \id$ to $\ket{\psi}^{AB}$ is the incoherent state $\op{0}{0}\otimes \mathbbm{1}/2$. This shows that the complete decohering power is $\mathfrak{D}_r(\Lambda) = 2$. On the other hand, the generalized decohering power is upper bounded as $\mathscr{D}(\Phi) \leq C\left(\op{+}{+}\right)$, regardless of the coherence quantifier, with equality achieved on the erasing channel $\Lambda$. We thus obtain $\mathscr{D}_r(\Lambda) = 1$ which is strictly smaller than $\mathfrak{D}_r(\Lambda)$, as claimed.

The above results show that complete decohering power can  in general exceed the generalized decohering power. In the next step we will provide an upper bound on the complete decohering power, which depends only on the dimension of the corresponding Hilbert space.
\begin{lemma} \label{lem:DecoheringPowerBound}
	Complete decohering power of a quantum channel $\Phi$ on the Hilbert space of dimension $d$ is bounded above as
	\begin{equation}
	\mathfrak{D}_r(\Phi) \leq 2\log_2 d.
	\end{equation}
    If we restrict the initial states to be separable, the upper bound will be the same with those of generalized decohering power:
    \begin{equation}
    \sup_k \max_{\rho_{\mathrm{sep}}^{AB}}\left\{ C_{r}\left(\rho_{\mathrm{sep}}^{AB}\right)-C_{r}\left(\Phi\otimes\id_{k}\left[\rho_{\mathrm{sep}}^{AB}\right]\right)\right\} \leq \log_2 d.
    \end{equation}
\end{lemma}
\begin{proof}
		For a given quantum channel $\Phi$ we define the bipartite channel $\Phi' = \Phi^A \otimes \mathbbm{1}^B$, and $d=d_A$ is the dimension of (the Hilbert space of) system $A$. Observe that $C_r\left(\Phi'\left[\rho^{AB}\right]\right) \geq C_r(\rho^B)$ from the fact that $C_r$ does not increase under partial trace. We obtain the following inequality:
		\begin{align}
		\mathfrak{D}_r(\Phi) \leq& \sup_{k}  \max_{\rho^{AB}} \left\{ C_r(\rho^{AB}) - C_r(\rho^{B}) \right\} \nonumber \\
		= &\sup_{k}  \max_{\rho^{AB}} \left\{ S\left(\Delta\left[\rho^{AB}\right]\right) - S\left(\rho^{AB}\right) - S\left(\Delta\left[\rho^B\right]\right) + S\left(\rho^B\right) \right\} \nonumber \\
        \leq & \sup_{k}  \max_{\rho^{AB}} \left\{S\left(\Delta\left[\rho^A\right]\right) - S\left(\rho^{AB}\right)  + S\left(\rho^B\right) \right\}\nonumber \\
        \leq & \max_{\rho^{A}} \left\{S\left(\Delta\left[\rho^A\right]\right) + S\left(\rho^A\right)\right\} \leq 2\log_2 d. \label{eq:DecoheringPowerBound}
		\end{align}
    	Here, we used the inequalities $S(\Delta[\rho^{AB}]) - S(\Delta[\rho^B]) \leq S(\Delta[\rho^A])$ and $-S(\rho^{AB}) + S(\rho^B) \leq S(\rho^A)$. 
    	When the state $\rho^{AB}$ is separable, it is also true that $-S(\rho^{AB}) + S(\rho^B) < 0 $. We apply this to the second line of above equation and obtain the upper bound for separable input states as follows. 
    	\begin{align}
    	\mathfrak{D}_r(\Phi) \leq & \sup_{k}  \max_{\rho_{\mathrm{sep}}^{AB}} \left\{S\left(\Delta\left[\rho_{\mathrm{sep}}^{AB}\right]\right) - S\left(\Delta\left[\rho^{B}\right]\right)\right\} \nonumber \\
        \leq & \max_{\rho^{A}} S\left(\Delta\left[\rho^A\right]\right) \leq \log_2 d. \label{eq:DecoheringPowerBoundSeparable}
    	\end{align} 	
\end{proof}
Note that the upper bounds in Lemma~\ref{lem:DecoheringPowerBound} are achieved for the erasing channel $\Lambda[\rho] = \ket{0}\!\bra{0}$ in either cases, entangled and separable inputs. This means that these upper bounds are the optimal bounds which depend on the dimension of the quantum channel only. The above proof also reveals interesting properties of quantum states achieving the maximal value of complete decohering power. In order to achieve the maximum value $2 \log_2 d$ for entangled inputs, the local state $\rho^A$ in Eq.~(\ref{eq:DecoheringPowerBound}) must satisfy $S(\Delta[\rho^A]) = S(\rho^A) = \log_2 d_A$, which means that $\rho^A$ must be maximally mixed. If, however, $\rho^A$ is maximally coherent, i.e., $S(\Delta[\rho^A]) = \log_2 d_A$ and $S(\rho^A) = 0$, then the eliminated coherence cannot be greater than $\log_2 d_A$ which is the maximum value of generalized decohering power.

The arguments just presented suggest that a quantum channel can eliminate more coherence when acting on one part of an entangled bipartite state. However, there are entangled input states where the amount of coherence eliminated through given $\Phi$ does not exceed $\mathscr{D}_r(\Phi)$. Also even if two different initial states have the same entanglement and coherence, their coherence after the application of the channel can be different. For example, consider the states $\ket{\psi}=\sin\theta\ket{0+}+\cos\theta\ket{1-}$, $\ket{\phi} = \sin\theta\ket{+0}+\cos\theta\ket{-1}$. Both states have the same entanglement and coherence but their coherence after the erasing channel $\Lambda[\rho]=\op{0}{0}$ are different.

The results presented above lead to an interesting question: if we use separable states in the definition of complete decohering power, does it coincide with the generalized decohering power? A rigorous proof of this statement would show that entanglement can enhance the ability of a quantum channel to destroy coherence, and provide another quantitative connection between the resource theories of entanglement and coherence. We leave this questions open for future research. 

\section{Coherence measures based on \texorpdfstring{$\ell_1$}{}-norm \label{sec:MatrixNorms}}  

We will now discuss cohering and decohering powers for the $\ell_1$-norm of coherence defined as~\cite{Baumgratz-2014a}
\begin{align}
C_{\ell_1}(\rho) = \min_{\sigma \in \mathcal{I}} ||\rho - \sigma||_{\ell_1}=  \sum_{i \neq j}|\bra{i}\rho \ket{j}|
\end{align}
with the $\ell_1$-norm $||M||_{\ell_1} = \sum_{i,j}|M_{ij}|$.
Maximum coherence is given by $C_{\ell_1}(\op{+_d}{+_d}) = d-1$. From this result, we immediately see that the generalized cohering power is at most $d-1$. As we will see in the following proposition, the complete cohering power is unbounded for almost all quantum channels.
\begin{proposition}\label{thm:ccpel}
For $\ell_1$-norm of coherence, complete cohering power of $\Phi$ whose generalized cohering power is not zero is unbounded. 
\end{proposition}
	\begin{proof}
		Let us consider a product state $\rho^{AB} = \rho^A \otimes \rho^B$ as the input for $\Phi$. Using the equality~\cite{Bu-2017a} $C_{\ell_1}(\rho^A \otimes \sigma^B) = [ C_{\ell_1}(\rho^A) + 1 ][ C_{\ell_1}(\sigma^B) + 1]-1$ we obtain:
		\begin{align*}
		\mathfrak{C}_{\ell_{1}}(\Phi) \geq & \left\{ C_{\ell_{1}}\left(\Phi\otimes\id_k\left[\rho^{AB}\right]\right)-C_{\ell_{1}}\left(\rho^{AB}\right)\right\}  \\
= & \bigg\{\left[C_{\ell_{1}}\left(\Phi\left[\rho^{A}\right]\right)+1\right]\left[C_{\ell_{1}}\left(\rho^{B}\right)+1\right] \\
- & \left[C_{\ell_{1}}\left(\rho^{A}\right)+1\right]\left[C_{\ell_{1}}\left(\rho^{B}\right)+1\right]\bigg\} \\
= & \left[C_{\ell_{1}}\left(\Phi\left[\rho^{A}\right]\right)-C_{\ell_{1}}\left(\rho^{A}\right)\right]\left[C_{\ell_{1}}\left(\rho^{B}\right)+1\right] .
\end{align*}
If we assume $\rho^B$ to be the maximally coherent state and take the dimension of system $B$ larger, then the generated $\ell_1$-norm of coherence increases as long as $C_{\ell_1}(\Phi [\rho^A]) - C_{\ell_1}(\rho^A)$ is not zero for some $\rho^A$, in other words $\mathscr{C}_{\ell_1}(\Phi) > 0$. So the complete cohering power is unbounded and not equal to generalized cohering power.
\end{proof}

By the same reasoning as in the proof of Proposition~\ref{thm:ccpel}, we see that the complete decohering power of $\Phi$ is unbounded, whenever $\Phi$ has nonzero generalized decohering power.

The two measures of coherence, $\ell_1$-norm of coherence and relative entropy of coherence, behave very differently with respect to cohering and decohering powers of quantum channels. For the $\ell_1$-norm of coherence, both complete cohering and decohering powers are unbounded for many channels. Thus, the relative entropy of coherence seems better suited for estimating the ability of quantum channels to create or destroy coherence. A similar effect has been observed for the geometric measure of quantum discord \cite{PhysRevLett.105.190502}, which can increase indefinitely depending on the attached ancillary system \cite{Piani-2012}.

\section{Conclusions}
In this work we have investigated optimal methods for establishing and destroying quantum coherence via quantum channels. We found that correlations with ancillary systems do not enhance the ability of a quantum channel to create coherence. This result significantly simplifies the analysis of several quantities related to coherence generation with quantum channels, including the coherence generating capacity~\cite{Ben-Dana-2017a}. On the other hand, we found that entanglement with an ancillary system can improve the ability of a channel to destroy quantum coherence. These results open the possibility that every entangled state can show a coherence decay above generalized cohering power for some quantum channel. Proving this statement would establish another rigorous and operationally meaningful connection between the resource theories of entanglement and coherence. We leave the proof of this statement for future research. 

Most of our analysis concerns coherence quantifiers defined via the relative entropy, which have an operational meaning for the resource theory of coherence in the asymptotic limit~\cite{Winter-2016a}. We have also investigated another commonly used coherence quantifier based on the $\ell_1$-norm. We found that many channels show an unbounded coherence generation for this quantifier, if ancillary systems are taken into account. These results suggest that coherence measures based on the relative entropy are more suitable to describe the potential of quantum channels to establish and destroy quantum coherence.




\section*{Acknowledgements}
We thank Seok Hyung Lie for pointing out an error in an earlier version of this manuscript. S.R. and A.S. acknowledge financial support by the ''Quantum Optical Technologies'' project, carried out within the International Research Agendas programme of the Foundation for Polish Science co-financed by the European Union under the European Regional Development Fund.

\bibliography{CCP}

\begin{thebibliography}{21}%
\makeatletter
\providecommand \@ifxundefined [1]{%
 \@ifx{#1\undefined}
}%
\providecommand \@ifnum [1]{%
 \ifnum #1\expandafter \@firstoftwo
 \else \expandafter \@secondoftwo
 \fi
}%
\providecommand \@ifx [1]{%
 \ifx #1\expandafter \@firstoftwo
 \else \expandafter \@secondoftwo
 \fi
}%
\providecommand \natexlab [1]{#1}%
\providecommand \enquote  [1]{``#1''}%
\providecommand \bibnamefont  [1]{#1}%
\providecommand \bibfnamefont [1]{#1}%
\providecommand \citenamefont [1]{#1}%
\providecommand \href@noop [0]{\@secondoftwo}%
\providecommand \href [0]{\begingroup \@sanitize@url \@href}%
\providecommand \@href[1]{\@@startlink{#1}\@@href}%
\providecommand \@@href[1]{\endgroup#1\@@endlink}%
\providecommand \@sanitize@url [0]{\catcode `\\12\catcode `\$12\catcode
  `\&12\catcode `\#12\catcode `\^12\catcode `\_12\catcode `\%12\relax}%
\providecommand \@@startlink[1]{}%
\providecommand \@@endlink[0]{}%
\providecommand \url  [0]{\begingroup\@sanitize@url \@url }%
\providecommand \@url [1]{\endgroup\@href {#1}{\urlprefix }}%
\providecommand \urlprefix  [0]{URL }%
\providecommand \Eprint [0]{\href }%
\providecommand \doibase [0]{https://doi.org/}%
\providecommand \selectlanguage [0]{\@gobble}%
\providecommand \bibinfo  [0]{\@secondoftwo}%
\providecommand \bibfield  [0]{\@secondoftwo}%
\providecommand \translation [1]{[#1]}%
\providecommand \BibitemOpen [0]{}%
\providecommand \bibitemStop [0]{}%
\providecommand \bibitemNoStop [0]{.\EOS\space}%
\providecommand \EOS [0]{\spacefactor3000\relax}%
\providecommand \BibitemShut  [1]{\csname bibitem#1\endcsname}%
\let\auto@bib@innerbib\@empty
\bibitem [{\citenamefont {Horodecki}\ \emph {et~al.}(2009)\citenamefont
  {Horodecki}, \citenamefont {Horodecki}, \citenamefont {Horodecki},\ and\
  \citenamefont {Horodecki}}]{Horodecki-2009a}%
  \BibitemOpen
  \bibfield  {author} {\bibinfo {author} {\bibfnamefont {R.}~\bibnamefont
  {Horodecki}}, \bibinfo {author} {\bibfnamefont {P.}~\bibnamefont
  {Horodecki}}, \bibinfo {author} {\bibfnamefont {M.}~\bibnamefont
  {Horodecki}},\ and\ \bibinfo {author} {\bibfnamefont {K.}~\bibnamefont
  {Horodecki}},\ }\bibfield  {title} {\bibinfo {title} {Quantum entanglement},\
  }\href {https://doi.org/10.1103/RevModPhys.81.865} {\bibfield  {journal}
  {\bibinfo  {journal} {Rev. Mod. Phys.}\ }\textbf {\bibinfo {volume} {81}},\
  \bibinfo {eid} {865} (\bibinfo {year} {2009})}\BibitemShut {NoStop}%
\bibitem [{\citenamefont {Streltsov}\ \emph {et~al.}(2017)\citenamefont
  {Streltsov}, \citenamefont {Adesso},\ and\ \citenamefont
  {Plenio}}]{Streltsov-2017a}%
  \BibitemOpen
  \bibfield  {author} {\bibinfo {author} {\bibfnamefont {A.}~\bibnamefont
  {Streltsov}}, \bibinfo {author} {\bibfnamefont {G.}~\bibnamefont {Adesso}},\
  and\ \bibinfo {author} {\bibfnamefont {M.~B.}\ \bibnamefont {Plenio}},\
  }\bibfield  {title} {\bibinfo {title} {{Colloquium: Quantum coherence as a
  resource}},\ }\href {https://doi.org/10.1103/RevModPhys.89.041003} {\bibfield
   {journal} {\bibinfo  {journal} {Rev. Mod. Phys.}\ }\textbf {\bibinfo
  {volume} {89}},\ \bibinfo {pages} {041003} (\bibinfo {year}
  {2017})}\BibitemShut {NoStop}%
\bibitem [{\citenamefont {Wu}\ \emph {et~al.}(2021)\citenamefont {Wu},
  \citenamefont {Streltsov}, \citenamefont {Regula}, \citenamefont {Xiang},
  \citenamefont {Li},\ and\ \citenamefont {Guo}}]{Wu2105.06854}%
  \BibitemOpen
  \bibfield  {author} {\bibinfo {author} {\bibfnamefont {K.-D.}\ \bibnamefont
  {Wu}}, \bibinfo {author} {\bibfnamefont {A.}~\bibnamefont {Streltsov}},
  \bibinfo {author} {\bibfnamefont {B.}~\bibnamefont {Regula}}, \bibinfo
  {author} {\bibfnamefont {G.-Y.}\ \bibnamefont {Xiang}}, \bibinfo {author}
  {\bibfnamefont {C.-F.}\ \bibnamefont {Li}},\ and\ \bibinfo {author}
  {\bibfnamefont {G.-C.}\ \bibnamefont {Guo}},\ }\bibfield  {title} {\bibinfo
  {title} {Experimental progress on quantum coherence: detection,
  quantification, and manipulation},\ }\href {https://arxiv.org/abs/2105.06854}
  {\bibfield  {journal} {\bibinfo  {journal} {arXiv:2105.06854}\ } (\bibinfo
  {year} {2021})}\BibitemShut {NoStop}%
\bibitem [{\citenamefont {Horodecki}\ \emph {et~al.}(2005)\citenamefont
  {Horodecki}, \citenamefont {Oppenheim},\ and\ \citenamefont
  {Winter}}]{Horodecki-2005a}%
  \BibitemOpen
  \bibfield  {author} {\bibinfo {author} {\bibfnamefont {M.}~\bibnamefont
  {Horodecki}}, \bibinfo {author} {\bibfnamefont {J.}~\bibnamefont
  {Oppenheim}},\ and\ \bibinfo {author} {\bibfnamefont {A.}~\bibnamefont
  {Winter}},\ }\bibfield  {title} {\bibinfo {title} {{Partial quantum
  information}},\ }\href {https://doi.org/10.1038/nature03909} {\bibfield
  {journal} {\bibinfo  {journal} {Nature}\ }\textbf {\bibinfo {volume} {436}},\
  \bibinfo {pages} {673} (\bibinfo {year} {2005})}\BibitemShut {NoStop}%
\bibitem [{\citenamefont {Streltsov}\ \emph {et~al.}(2016)\citenamefont
  {Streltsov}, \citenamefont {Chitambar}, \citenamefont {Rana}, \citenamefont
  {Bera}, \citenamefont {Winter},\ and\ \citenamefont
  {Lewenstein}}]{Streltsov-2016a}%
  \BibitemOpen
  \bibfield  {author} {\bibinfo {author} {\bibfnamefont {A.}~\bibnamefont
  {Streltsov}}, \bibinfo {author} {\bibfnamefont {E.}~\bibnamefont
  {Chitambar}}, \bibinfo {author} {\bibfnamefont {S.}~\bibnamefont {Rana}},
  \bibinfo {author} {\bibfnamefont {M.~N.}\ \bibnamefont {Bera}}, \bibinfo
  {author} {\bibfnamefont {A.}~\bibnamefont {Winter}},\ and\ \bibinfo {author}
  {\bibfnamefont {M.}~\bibnamefont {Lewenstein}},\ }\bibfield  {title}
  {\bibinfo {title} {{Entanglement and Coherence in Quantum State Merging}},\
  }\href {https://doi.org/10.1103/PhysRevLett.116.240405} {\bibfield  {journal}
  {\bibinfo  {journal} {Phys. Rev. Lett.}\ }\textbf {\bibinfo {volume} {116}},\
  \bibinfo {pages} {240405} (\bibinfo {year} {2016})}\BibitemShut {NoStop}%
\bibitem [{\citenamefont {Aberg}(2006)}]{Aberg-2006a}%
  \BibitemOpen
  \bibfield  {author} {\bibinfo {author} {\bibfnamefont {J.}~\bibnamefont
  {Aberg}},\ }\bibfield  {title} {\bibinfo {title} {{Quantifying
  Superposition}},\ }\href {https://arxiv.org/abs/quant-ph/0612146} {\bibfield
  {journal} {\bibinfo  {journal} {arXiv:quant-ph/0612146}\ } (\bibinfo {year}
  {2006})}\BibitemShut {NoStop}%
\bibitem [{\citenamefont {Baumgratz}\ \emph {et~al.}(2014)\citenamefont
  {Baumgratz}, \citenamefont {Cramer},\ and\ \citenamefont
  {Plenio}}]{Baumgratz-2014a}%
  \BibitemOpen
  \bibfield  {author} {\bibinfo {author} {\bibfnamefont {T.}~\bibnamefont
  {Baumgratz}}, \bibinfo {author} {\bibfnamefont {M.}~\bibnamefont {Cramer}},\
  and\ \bibinfo {author} {\bibfnamefont {M.~B.}\ \bibnamefont {Plenio}},\
  }\bibfield  {title} {\bibinfo {title} {{Quantifying Coherence}},\ }\href
  {https://doi.org/10.1103/PhysRevLett.113.140401} {\bibfield  {journal}
  {\bibinfo  {journal} {Phys. Rev. Lett.}\ }\textbf {\bibinfo {volume} {113}},\
  \bibinfo {pages} {140401} (\bibinfo {year} {2014})}\BibitemShut {NoStop}%
\bibitem [{\citenamefont {Winter}\ and\ \citenamefont
  {Yang}(2016)}]{Winter-2016a}%
  \BibitemOpen
  \bibfield  {author} {\bibinfo {author} {\bibfnamefont {A.}~\bibnamefont
  {Winter}}\ and\ \bibinfo {author} {\bibfnamefont {D.}~\bibnamefont {Yang}},\
  }\bibfield  {title} {\bibinfo {title} {{Operational Resource Theory of
  Coherence}},\ }\href {https://doi.org/10.1103/PhysRevLett.116.120404}
  {\bibfield  {journal} {\bibinfo  {journal} {Phys. Rev. Lett.}\ }\textbf
  {\bibinfo {volume} {116}},\ \bibinfo {pages} {120404} (\bibinfo {year}
  {2016})}\BibitemShut {NoStop}%
\bibitem [{\citenamefont {Chitambar}\ and\ \citenamefont
  {Hsieh}(2016)}]{Chitambar-2016a}%
  \BibitemOpen
  \bibfield  {author} {\bibinfo {author} {\bibfnamefont {E.}~\bibnamefont
  {Chitambar}}\ and\ \bibinfo {author} {\bibfnamefont {M.-H.}\ \bibnamefont
  {Hsieh}},\ }\bibfield  {title} {\bibinfo {title} {{Relating the Resource
  Theories of Entanglement and Quantum Coherence}},\ }\href
  {https://doi.org/10.1103/PhysRevLett.117.020402} {\bibfield  {journal}
  {\bibinfo  {journal} {Phys. Rev. Lett.}\ }\textbf {\bibinfo {volume} {117}},\
  \bibinfo {pages} {020402} (\bibinfo {year} {2016})}\BibitemShut {NoStop}%
\bibitem [{\citenamefont {Ben~Dana}\ \emph
  {et~al.}(2017{\natexlab{a}})\citenamefont {Ben~Dana}, \citenamefont
  {Garc\'{\i}a~D\'{\i}az}, \citenamefont {Mejatty},\ and\ \citenamefont
  {Winter}}]{Ben-Dana-2017a}%
  \BibitemOpen
  \bibfield  {author} {\bibinfo {author} {\bibfnamefont {K.}~\bibnamefont
  {Ben~Dana}}, \bibinfo {author} {\bibfnamefont {M.}~\bibnamefont
  {Garc\'{\i}a~D\'{\i}az}}, \bibinfo {author} {\bibfnamefont {M.}~\bibnamefont
  {Mejatty}},\ and\ \bibinfo {author} {\bibfnamefont {A.}~\bibnamefont
  {Winter}},\ }\bibfield  {title} {\bibinfo {title} {{Resource theory of
  coherence: Beyond states}},\ }\href
  {https://doi.org/10.1103/PhysRevA.95.062327} {\bibfield  {journal} {\bibinfo
  {journal} {Phys. Rev. A}\ }\textbf {\bibinfo {volume} {95}},\ \bibinfo
  {pages} {062327} (\bibinfo {year} {2017}{\natexlab{a}})}\BibitemShut
  {NoStop}%
\bibitem [{\citenamefont {Ben~Dana}\ \emph
  {et~al.}(2017{\natexlab{b}})\citenamefont {Ben~Dana}, \citenamefont
  {Garc\'{\i}a~D\'{\i}az}, \citenamefont {Mejatty},\ and\ \citenamefont
  {Winter}}]{BenDanaPhysRevA.96.059903}%
  \BibitemOpen
  \bibfield  {author} {\bibinfo {author} {\bibfnamefont {K.}~\bibnamefont
  {Ben~Dana}}, \bibinfo {author} {\bibfnamefont {M.}~\bibnamefont
  {Garc\'{\i}a~D\'{\i}az}}, \bibinfo {author} {\bibfnamefont {M.}~\bibnamefont
  {Mejatty}},\ and\ \bibinfo {author} {\bibfnamefont {A.}~\bibnamefont
  {Winter}},\ }\bibfield  {title} {\bibinfo {title} {{Erratum: Resource theory
  of coherence: Beyond states [Phys. Rev. A 95, 062327 (2017)]}},\ }\href
  {https://doi.org/10.1103/PhysRevA.96.059903} {\bibfield  {journal} {\bibinfo
  {journal} {Phys. Rev. A}\ }\textbf {\bibinfo {volume} {96}},\ \bibinfo
  {pages} {059903} (\bibinfo {year} {2017}{\natexlab{b}})}\BibitemShut
  {NoStop}%
\bibitem [{\citenamefont {Mani}\ and\ \citenamefont
  {Karimipour}(2015)}]{Mani-2015a}%
  \BibitemOpen
  \bibfield  {author} {\bibinfo {author} {\bibfnamefont {A.}~\bibnamefont
  {Mani}}\ and\ \bibinfo {author} {\bibfnamefont {V.}~\bibnamefont
  {Karimipour}},\ }\bibfield  {title} {\bibinfo {title} {Cohering and
  decohering power of quantum channels},\ }\href
  {https://doi.org/10.1103/PhysRevA.92.032331} {\bibfield  {journal} {\bibinfo
  {journal} {Phys. Rev. A}\ }\textbf {\bibinfo {volume} {92}},\ \bibinfo
  {pages} {032331} (\bibinfo {year} {2015})}\BibitemShut {NoStop}%
\bibitem [{\citenamefont {Yao}\ \emph {et~al.}(2015)\citenamefont {Yao},
  \citenamefont {Xiao}, \citenamefont {Ge},\ and\ \citenamefont
  {Sun}}]{Yao-2015a}%
  \BibitemOpen
  \bibfield  {author} {\bibinfo {author} {\bibfnamefont {Y.}~\bibnamefont
  {Yao}}, \bibinfo {author} {\bibfnamefont {X.}~\bibnamefont {Xiao}}, \bibinfo
  {author} {\bibfnamefont {L.}~\bibnamefont {Ge}},\ and\ \bibinfo {author}
  {\bibfnamefont {C.~P.}\ \bibnamefont {Sun}},\ }\bibfield  {title} {\bibinfo
  {title} {Quantum coherence in multipartite systems},\ }\href
  {https://doi.org/10.1103/PhysRevA.92.022112} {\bibfield  {journal} {\bibinfo
  {journal} {Phys. Rev. A}\ }\textbf {\bibinfo {volume} {92}},\ \bibinfo
  {pages} {022112} (\bibinfo {year} {2015})}\BibitemShut {NoStop}%
\bibitem [{\citenamefont {Bu}\ and\ \citenamefont {Xiong}(2017)}]{Bu-2017b}%
  \BibitemOpen
  \bibfield  {author} {\bibinfo {author} {\bibfnamefont {K.}~\bibnamefont
  {Bu}}\ and\ \bibinfo {author} {\bibfnamefont {C.}~\bibnamefont {Xiong}},\
  }\bibfield  {title} {\bibinfo {title} {A note on cohering power and
  de-cohering power},\ }\href {https://doi.org/10.26421/QIC17.13-14-8}
  {\bibfield  {journal} {\bibinfo  {journal} {Quant. Inf. Comp.}\ }\textbf
  {\bibinfo {volume} {17}},\ \bibinfo {pages} {1206} (\bibinfo {year}
  {2017})}\BibitemShut {NoStop}%
\bibitem [{\citenamefont {Garc\'{\i}a-D\'{\i}az}\ \emph
  {et~al.}(2016)\citenamefont {Garc\'{\i}a-D\'{\i}az}, \citenamefont {Egloff},\
  and\ \citenamefont {Plenio}}]{Garcia-Diaz-2016a}%
  \BibitemOpen
  \bibfield  {author} {\bibinfo {author} {\bibfnamefont {M.}~\bibnamefont
  {Garc\'{\i}a-D\'{\i}az}}, \bibinfo {author} {\bibfnamefont {D.}~\bibnamefont
  {Egloff}},\ and\ \bibinfo {author} {\bibfnamefont {M.~B.}\ \bibnamefont
  {Plenio}},\ }\bibfield  {title} {\bibinfo {title} {A note on coherence power
  of n-dimensional unitary operators},\ }\href
  {https://doi.org/https://doi.org/10.26421/QIC16.15-16} {\bibfield  {journal}
  {\bibinfo  {journal} {Quant. Inf. Comp.}\ }\textbf {\bibinfo {volume} {16}},\
  \bibinfo {pages} {1282} (\bibinfo {year} {2016})}\BibitemShut {NoStop}%
\bibitem [{\citenamefont {Bu}\ \emph {et~al.}(2017)\citenamefont {Bu},
  \citenamefont {Kumar}, \citenamefont {Zhang},\ and\ \citenamefont
  {Wu}}]{Bu-2017a}%
  \BibitemOpen
  \bibfield  {author} {\bibinfo {author} {\bibfnamefont {K.}~\bibnamefont
  {Bu}}, \bibinfo {author} {\bibfnamefont {A.}~\bibnamefont {Kumar}}, \bibinfo
  {author} {\bibfnamefont {L.}~\bibnamefont {Zhang}},\ and\ \bibinfo {author}
  {\bibfnamefont {J.}~\bibnamefont {Wu}},\ }\bibfield  {title} {\bibinfo
  {title} {Cohering power of quantum operations},\ }\href
  {https://doi.org/10.1016/j.physleta.2017.03.022} {\bibfield  {journal}
  {\bibinfo  {journal} {Phys. Lett. A}\ }\textbf {\bibinfo {volume} {381}},\
  \bibinfo {pages} {1670 } (\bibinfo {year} {2017})}\BibitemShut {NoStop}%
\bibitem [{\citenamefont {Streltsov}\ \emph {et~al.}(2015)\citenamefont
  {Streltsov}, \citenamefont {Singh}, \citenamefont {Dhar}, \citenamefont
  {Bera},\ and\ \citenamefont {Adesso}}]{Streltsov-2015a}%
  \BibitemOpen
  \bibfield  {author} {\bibinfo {author} {\bibfnamefont {A.}~\bibnamefont
  {Streltsov}}, \bibinfo {author} {\bibfnamefont {U.}~\bibnamefont {Singh}},
  \bibinfo {author} {\bibfnamefont {H.~S.}\ \bibnamefont {Dhar}}, \bibinfo
  {author} {\bibfnamefont {M.~N.}\ \bibnamefont {Bera}},\ and\ \bibinfo
  {author} {\bibfnamefont {G.}~\bibnamefont {Adesso}},\ }\bibfield  {title}
  {\bibinfo {title} {{Measuring Quantum Coherence with Entanglement}},\ }\href
  {https://doi.org/10.1103/PhysRevLett.115.020403} {\bibfield  {journal}
  {\bibinfo  {journal} {Phys. Rev. Lett.}\ }\textbf {\bibinfo {volume} {115}},\
  \bibinfo {pages} {020403} (\bibinfo {year} {2015})}\BibitemShut {NoStop}%
\bibitem [{\citenamefont {Takahashi}\ and\ \citenamefont
  {Chitambar}(2018)}]{Takahashi-2018}%
  \BibitemOpen
  \bibfield  {author} {\bibinfo {author} {\bibfnamefont {M.}~\bibnamefont
  {Takahashi}}\ and\ \bibinfo {author} {\bibfnamefont {E.}~\bibnamefont
  {Chitambar}},\ }\bibfield  {title} {\bibinfo {title} {Comparing coherence and
  entanglement under resource non-generating unitary transformations},\ }\href
  {https://doi.org/10.1088/1751-8121/aacc5c} {\bibfield  {journal} {\bibinfo
  {journal} {J. Phys. A}\ }\textbf {\bibinfo {volume} {51}},\ \bibinfo {pages}
  {414003} (\bibinfo {year} {2018})}\BibitemShut {NoStop}%
\bibitem [{\citenamefont {D{\'{i}}az}\ \emph {et~al.}(2018)\citenamefont
  {D{\'{i}}az}, \citenamefont {Fang}, \citenamefont {Wang}, \citenamefont
  {Rosati}, \citenamefont {Skotiniotis}, \citenamefont {Calsamiglia},\ and\
  \citenamefont {Winter}}]{Garcia-Diaz-2018a}%
  \BibitemOpen
  \bibfield  {author} {\bibinfo {author} {\bibfnamefont {M.~G.}\ \bibnamefont
  {D{\'{i}}az}}, \bibinfo {author} {\bibfnamefont {K.}~\bibnamefont {Fang}},
  \bibinfo {author} {\bibfnamefont {X.}~\bibnamefont {Wang}}, \bibinfo {author}
  {\bibfnamefont {M.}~\bibnamefont {Rosati}}, \bibinfo {author} {\bibfnamefont
  {M.}~\bibnamefont {Skotiniotis}}, \bibinfo {author} {\bibfnamefont
  {J.}~\bibnamefont {Calsamiglia}},\ and\ \bibinfo {author} {\bibfnamefont
  {A.}~\bibnamefont {Winter}},\ }\bibfield  {title} {\bibinfo {title} {Using
  and reusing coherence to realize quantum processes},\ }\href
  {https://doi.org/10.22331/q-2018-10-19-100} {\bibfield  {journal} {\bibinfo
  {journal} {{Quantum}}\ }\textbf {\bibinfo {volume} {2}},\ \bibinfo {pages}
  {100} (\bibinfo {year} {2018})}\BibitemShut {NoStop}%
\bibitem [{\citenamefont {Daki\ifmmode~\acute{c}\else \'{c}\fi{}}\ \emph
  {et~al.}(2010)\citenamefont {Daki\ifmmode~\acute{c}\else \'{c}\fi{}},
  \citenamefont {Vedral},\ and\ \citenamefont
  {Brukner}}]{PhysRevLett.105.190502}%
  \BibitemOpen
  \bibfield  {author} {\bibinfo {author} {\bibfnamefont {B.}~\bibnamefont
  {Daki\ifmmode~\acute{c}\else \'{c}\fi{}}}, \bibinfo {author} {\bibfnamefont
  {V.}~\bibnamefont {Vedral}},\ and\ \bibinfo {author} {\bibfnamefont
  {C.}~\bibnamefont {Brukner}},\ }\bibfield  {title} {\bibinfo {title}
  {Necessary and sufficient condition for nonzero quantum discord},\ }\href
  {https://doi.org/10.1103/PhysRevLett.105.190502} {\bibfield  {journal}
  {\bibinfo  {journal} {Phys. Rev. Lett.}\ }\textbf {\bibinfo {volume} {105}},\
  \bibinfo {pages} {190502} (\bibinfo {year} {2010})}\BibitemShut {NoStop}%
\bibitem [{\citenamefont {Piani}(2012)}]{Piani-2012}%
  \BibitemOpen
  \bibfield  {author} {\bibinfo {author} {\bibfnamefont {M.}~\bibnamefont
  {Piani}},\ }\bibfield  {title} {\bibinfo {title} {Problem with geometric
  discord},\ }\href {https://doi.org/10.1103/PhysRevA.86.034101} {\bibfield
  {journal} {\bibinfo  {journal} {Phys. Rev. A}\ }\textbf {\bibinfo {volume}
  {86}},\ \bibinfo {pages} {034101} (\bibinfo {year} {2012})}\BibitemShut
  {NoStop}%
\end{thebibliography}%

\end{document}